\documentclass[journal,12pt,onecolumn,draftclsnofoot]{IEEEtran}

\usepackage{amsfonts}
\usepackage{amssymb}
\usepackage[bookmarks=false, draft]{hyperref}
\usepackage{graphicx}
\usepackage{subfigure}
\usepackage{enumerate}
\usepackage{amsmath}
\usepackage{color}
\usepackage{amsthm}
\usepackage{amsmath}
\usepackage{algorithm}
\usepackage{algpseudocode}
\usepackage{breakurl}
\hypersetup{hidelinks}

\newcommand{\bp}{\begin{proof} \small }
\newcommand{\ep}{\end{proof} \normalsize}
\newcommand{\epx}{\end{proof} \small}
\newcommand{\bpa}{\begin{proofappx} \footnotesize }
\newcommand{\epa}{\end{proofappx} \small }
\newtheorem{theorem}{Theorem}

\newtheorem{definition}{Definition}

\newtheorem*{theorem*}{Theorem}
\newtheorem*{proposition*}{Proposition}
\newtheorem*{corollary*}{Corollary}
\newtheorem*{lemma*}{Lemma}
\newtheorem*{assumption*}{Assumption}
\newtheorem*{definition*}{Definition}
\newtheorem*{claim*}{Claim}

\newcommand{\bm}[1]{\mbox{\boldmath $#1$}}

\newcommand{\be}{\begin{equation}}
\newcommand{\ee}{\end{equation}}
\newcommand{\bs}{\begin{subequations}}
\newcommand{\es}{\end{subequations}}
\newcommand{\bq}{\begin{eqnarray}}
\newcommand{\eq}{\end{eqnarray}}
\newcommand{\bqn}{\begin{eqnarray*}}
\newcommand{\eqn}{\end{eqnarray*}}

\newcommand{\ba}{\left[ \begin{array}}
\newcommand{\ea}{\\ \end{array} \right]}
\newcommand{\ben}{\begin{enumerate}}
\newcommand{\een}{\end{enumerate}}

\def\real{{\mathchoice%
{\hbox{\rm\setbox1=\hbox{I}\copy1\kern-.45\wd1 R}}
{\hbox{\rm\setbox1=\hbox{I}\copy1\kern-.45\wd1 R}}
{\hbox{\scriptsize\rm\setbox1=\hbox{I}\copy1\kern-.45\wd1 R}}
{\hbox{\scriptsize\rm\setbox1=\hbox{I}\copy1\kern-.45\wd1 R}}}}

\def\Zint{{\mathchoice{\setbox1=\hbox{\sf Z}\copy1\kern-.75\wd1\box1}
{\setbox1=\hbox{\sf Z}\copy1\kern-.75\wd1\box1}
{\setbox1=\hbox{\scriptsize\sf Z}\copy1\kern-.75\wd1\box1}
{\setbox1=\hbox{\scriptsize\sf Z}\copy1\kern-.75\wd1\box1}}}
\newcommand{\complex}{ \hbox{\rm C\kern-0.45em\rule[.07em]{.02em}{.58em}%
\kern 0.43em}}

\begin{document}
\title{Socially Trusted Collaborative Edge Computing \\in Ultra Dense Networks}

\author{Lixing~Chen,~\IEEEmembership{Student~Member,~IEEE,}
        Jie~Xu,~\IEEEmembership{Member,~IEEE}

\thanks{L. Chen and J. Xu are with the Department of Electrical and
	Computer Engineering, University of Miami, USA. Email: lx.chen@miami.edu, jiexu@miami.edu.}
}
\maketitle

\begin{abstract}
Small cell base stations (SBSs) endowed with cloud-like computing capabilities are considered as a key enabler of edge computing (EC), which provides ultra-low latency and location-awareness for a variety of emerging mobile applications and the Internet of Things. However, due to the limited computation resources of an individual SBS, providing computation services of high quality to its users faces significant challenges when it is overloaded with an excessive amount of computation workload. In this paper, we propose collaborative edge computing among SBSs by forming SBS coalitions to share computation resources with each other, thereby accommodating more computation workload in the edge system and reducing reliance on the remote cloud. A novel SBS coalition formation algorithm is developed based on the coalitional game theory to cope with various new challenges in small-cell-based edge systems, including the co-provisioning of radio access and computing services, cooperation incentives, and potential security risks. To address these challenges, the proposed method (1) allows collaboration at both the user-SBS association stage and the SBS peer offloading stage by exploiting the ultra dense deployment of SBSs, (2) develops a payment-based incentive mechanism that implements proportionally fair utility division to form stable SBS coalitions, and (3) builds a social trust network for managing security risks among SBSs due to collaboration. Systematic simulations in practical scenarios are carried out to evaluate the efficacy and performance of the proposed method, which shows that tremendous edge computing performance improvement can be achieved.
\end{abstract}

\section{Introduction}
Pervasive mobile computing and the Internet of Things are driving the development of many new applications that are both compute-demanding and latency-sensitive, such as cognitive assistance, mobile gaming and augmented reality. Although cloud computing enables convenient access to a centralized pool of configurable and powerful computing resources, it often cannot meet the stringent requirements of latency-sensitive applications due to the often unpredictable network latency and expensive bandwidth \cite{mao2017mobile, shi2016edge}. The growing amount of distributed data further makes it impractical or resource-prohibitive to transport all the data over today's already-congested backbone networks to the remote cloud \cite{rivera2014gartner}. As a remedy to these limitations, Edge Computing (EC) \cite{lopez2015edge} emerges as a new computing paradigm to push the frontier of computing applications, data, and services away from centralized cloud computing infrastructures to the logical extremes of a network thereby enabling analytics and knowledge generation to occur at the data source.

\begin{figure}[t]
	\centering	
	\includegraphics[width=3.5 in]{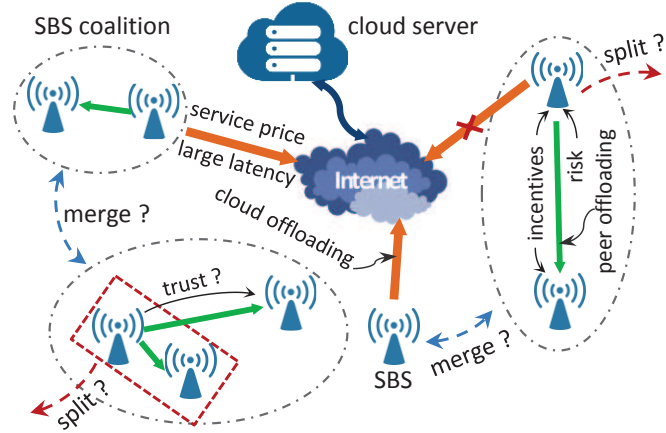}
	\caption{Illustration of socially trusted collaborative EC}
	\label{scenario}
	\vspace{-0.2 in}
\end{figure}

Considered as a key enabler of EC, small cell base stations (SBSs), such as femtocells and picocells, endowed with cloud-like computing and storage capability can serve end-users' computation requests as a substitute of the cloud while providing LTE connectivity to the Internet, especially for indoor premises. Such migration between computation and wireless communication service provisioning at the network edge was envisioned in the TROPIC project \cite{TROPIC}. Nonetheless, compared with mega-scale data centers, SBS will be limited in computing resources. Therefore, computation workload exceeding the SBS's computing capacity still has to be sent to the cloud, resulting in a hierarchical offloading structure among end-users, SBSs, and the cloud. Although there have been a few works \cite{xu2016online, zhao2015cooperative} studying the optimization of the offloading strategy in this hierarchical structure, relying on a single SBS significantly limits the EC performance. Fortunately, the ultra dense deployment of SBSs in the next generation (5G) mobile networks \cite{andrews2014will} creates an opportunity for nearby SBSs to collaboratively form a computation resource pool. Such a collaborative EC infrastructure, also known as Femto-Cloud \cite{tanzil2016distributed}, improves the efficiency of system resource utilization by exploiting the spatial diversity of workload patterns, thereby significantly enhancing EC performance. For instance, a cluster of SBSs can coordinate among themselves to serve computation requests by transferring workload from overloaded SBSs to nearby SBSs with a light workload.

The idea of balancing computation workload among nearby SBSs is similar to geographical load balancing \cite{lin2012online, lou2015spatio} in data center networks, which has been extensively studied. However, edge computing in ultra dense networks faces many new challenges of forming an effective collaborative network. First, whereas conventional clouds manage only the computing resource, moving the computing resource to the network edge leads to the co-provisioning of radio access and computing services by the SBSs, thus mandating a new model for understanding the interdependency between the management of the two resources. Unlike cloud computing that is agnostic to user location, the physical association between the SBSs and the end users becomes a critical design aspect that has a significant impact on the computing performance. Second, a distinct feature of SBSs is that they are often owned and deployed by individual users (e.g. home/enterprise owners). Unlike macro base stations that are deployed by the network operator, the operator has only minimum control over SBSs. Without proper incentives, SBSs will be reluctant to participate in the collaborative edge computing process. Therefore, incentives must be devised and incorporated into the load balancing scheme. Of particular importance is the design of a trust management module that takes into account the social trust relationship between the SBSs to minimize the security and privacy risks in collaboration.

In this paper, we design a socially trusted collaborative edge computing platform for ultra dense networks (see Figure \ref{scenario} for illustration). (1) We use payment as the incentive mechanism for individual SBSs to collaborate. Specifically, overloaded SBSs can pay nearby SBSs with spare computing resources to process their workload instead of offloading it to the remote cloud.  The collaborative network formation and the associated payment scheme are designed under the coalitional game theoretic framework and we prove that our proposed scheme results in the optimal stable coalition among the SBSs.  (2) When forming the coalition among SBSs, in addition to workload balancing at the SBS level, we allow end-users to directly switch their association to neighboring SBSs by exploiting the ultra dense deployment of SBSs, thereby avoiding transmission energy and latency incurred in the workload transfer between SBSs. This mixed workload balancing scheme is in stark contrast with geographical load balancing in data center networks or collaborative Femto-Cloud \cite{tanzil2016distributed} that does not consider the ultra dense deployment. (3) The proposed platform has a dedicated trust management component that manages the trust between any two SBSs to enable security-aware collaboration. Apart from the physical network of SBSs, we construct a social trust network that characterizes the trust between SBSs, which can be used to determine the security measure to be put on the processing/offloading of computation workload from/to neighbor SBSs, thereby reducing security and privacy leakage risks. (4) We conduct extensive systematic simulation studies to evaluate our proposed scheme in practical settings and understand how different SBS coalitions are formed and when collaborative EC is the most useful. Our results show that the proposed collaborative EC method can reduce the system cost by more than 40\% compared to standalone EC systems without collaboration.

The rest of this paper is organized as follows. Section \ref{sec_sys_model} presents the system model. Section \ref{sec_coalition_game} develops a coalitional game for cooperative SBS network. Section \ref{sec_CF} proposes a merge-split framework for distributed coalition formation. Simulations are carried out in Section \ref{sec_simulation}. Section \ref{sec_related_work} reviews related works, followed by the conclusion in Section \ref{sec_conclusion}.

\section{System Model}\label{sec_sys_model}
We consider $N$ SBSs, indexed by $\mathcal{N}=\{1,2,...,N\}$, endowed with heterogeneous computing capabilities. SBSs are located in separate rooms, possibly on different floors, in a multi-story building. These SBSs have Internet access and therefore can offload computation tasks to the remote cloud when its own computational capacity cannot accommodate the demand. Let $\mathcal{M} = \{1,2,...,M\}$ denote the set of all mobile user equipments (MUEs) in the building. Each SBS has a set of authorized MUEs, denoted by $\mathcal{M}_i\subseteq \mathcal{M}$. For instance, MUEs (e.g. mobile phones, laptops etc.) of employees in a business are authorized to access the communication/computing service of the SBS deployed by the business. MUEs of family members can also access its home SBS. However, due to the ultra dense deployment of SBSs in the building, an MUE is in the radio coverage of multiple SBSs and hence, there is a potential for SBSs to collaboratively balance the workload via direct MUE-SBS association.

We consider a time-slotted system. In each time slot, each MUE has a certain amount of computationally intensive tasks that need be offloaded to either the SBSs or the cloud for processing. We will focus on the problem in one time slot. There are two types of computation tasks: private tasks and normal tasks. Private tasks of an MUE must be processed by its own SBS or the cloud (which is assumed to be secure). They cannot be processed by other SBSs due to privacy concerns since SBSs are deployed by individual home/business owners which may not be fully trusted and less protected than the cloud. Normal tasks are less security-sensitive and hence, they can be processed by either its own SBS, the secure cloud or other nearby SBSs. Therefore, the task requests from MUE $m\in\mathcal{M}$ in a time slot are described by a tuple $(\lambda^a_m,\tau_m)$ where $\lambda^a_m$ is rate of task arrival in the current time slot (assuming a Poisson arrival process) and $\tau_m\in[0,1]$ is the fraction of private tasks. Without SBS collaboration, the total task arrival rate to SBS $i$ from all its authorized MUEs is thus:
\begin{align}
	\lambda^s_i=\sum_{m\in\mathcal{M}_i}\lambda^a_m
\end{align}

Due to the limited computational capacity of SBS $i$, there is a maximum task arrival rate $\omega^{\max}_i$ that SBS $i$ can handle. We define $\alpha_i \triangleq\omega^{\max}_i-\lambda^s_i$ as the computing resource surplus (or deficit) of SBS $i$. If $\alpha_i \geq 0$, then SBS $i$ has spare computing resources that it can share with other SBSs. If $\alpha_i < 0$, then SBS $i$ needs to acquire additional computing resource from either the cloud or other peer SBSs to meet its demand. Therefore, there is a potential computing resource exchange market among the SBSs. We call SBSs with $\alpha \geq 0$ potential ``sellers'' and those with $\alpha < 0$ potential ``buyers''.

\subsection{Overview of Collaborative Edge Computing}
Buyer SBSs have the following two options to acquire additional computing resources:

\textbf{SBS-to-Cloud offloading}: The SBS can further offload the unsatisfied computation tasks to the remote cloud. However, there will be extra transmission delay costs due to the large round-trip time to the remote cloud. Moreover, the cloud will also charge the SBS service fees for using the cloud computing resources.

\textbf{SBS coalitions}: Instead of relying on the remote cloud, SBSs can collaborate with each other by forming edge computing coalitions, in order to reduce the number of computation tasks offloaded to the cloud, thereby improving the delay performance and cutting their expenses on using the cloud service. Specifically, there are two ways to collaborate:
\begin{itemize}
	\item \emph{SBS peer offloading}: a buyer SBS first receives the computation tasks from its authorized MUEs and then further transmits via the wireless link some of the received tasks to nearby seller SBSs for processing.
	\item \emph{MUE-SBS association}: if an authorized MUE is also in the coverage of the seller SBS, a buyer SBS can associate this MUE to the seller SBS and directly offload its computation tasks to that seller SBS. This method exploits the dense deployment of SBSs and further saves the task transmission delay and energy cost.
\end{itemize}

Clearly, load balancing via MUE-SBS association is preferred due to the reduced overhead cost. Therefore, in SBS coalitions, MUE-SBS association is adopted with a higher priority. When load balancing cannot be realized via MUE-SBS association, SBS peer offloading is then executed. Typically, these two methods are used at the same time since MUEs are distributed in the network and not all MUEs are covered by seller SBSs.

\begin{figure}[htb]
	\centering	
	\includegraphics[width=3.5 in]{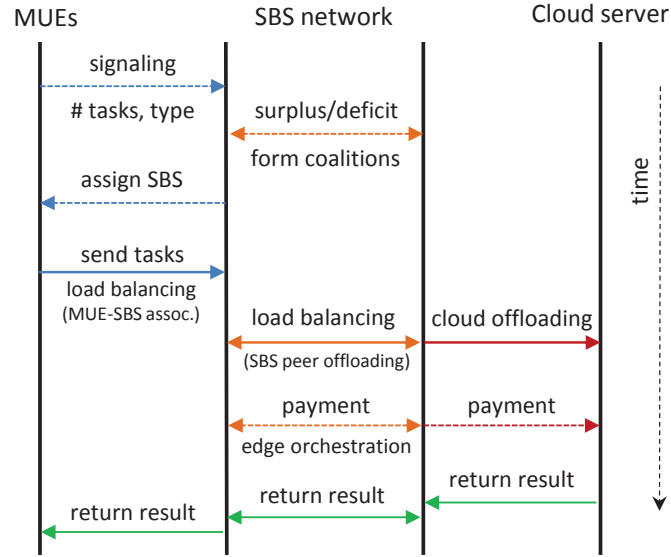}
	\caption{Operating time sequence diagram of collaborative edge computing}
	\label{operating_time_sequence}
\end{figure}

Different coalitions (i.e. seller-buyer matchings) lead to different edge computing performance. Since SBSs are self-interested and do not provide computing services to other SBSs for free, a key issue that this paper will address is how to design payment mechanisms to provide seller SBSs with incentives to cooperate and form the optimal stable coalitions. Before we design the coalition mechanism, we first show a time sequence diagram in Figure \ref{operating_time_sequence} to illustrate the operation of the collaborative edge computing system. At the beginning of each operational slot, a signaling phase is introduced in which MUEs report task requests. Based on these requests, SBSs form edge computing coalitions by playing a distributed coalitional game, taking into consideration of the local processing cost, cloud offloading cost, SBS peer offloading cost and collaboration risk. Buyer SBSs make payments to seller SBSs for using their computing services. An edge orchestrator, which is a trusted third party, is introduced to facilitate the payment process. Specifically, the buyer SBSs first send payments to the edge orchestrator, which then distributes the payments to the seller SBSs. The time slot ends with returning computation results to the MUEs.

Next, we model the various costs involved in the collaborative edge computing system. For each SBS, its total cost comprises two main parts: operational cost and risk management cost. These costs form the basis of the coalition formation game.

\subsection{Operational Cost}
We first model the operational cost of the edge computing coalitions incurred in different stages of the system.

\textbf{Stage I: MUE-SBS association}. We consider that an SBS absorbs the association cost (i.e. MUE-to-SBS transmission delay and energy consumption) of its authorized MUEs as part of its operational cost. Let $\mathcal{M}_{ij}\subseteq\mathcal{M}_{i} (\forall j\in\mathcal{N}, j\neq i$) denote the set of authorized MUEs of SBS $i$ that are associated to SBS $j$. The achievable transmission rate $r^a_{mi}$ between MUE $m$ and SBS $i$ is given by the Shannon capacity
\begin{align}
r^a_{mi} = W\log\left(1 + \frac{p^a_{mi} H_{mi}}{\sigma^2 + I}\right)
\end{align}
where $W$ is the channel bandwidth, $H_{mi}$  is the channel gain between SBS $i$ and user $m$, $\sigma^2$ is the noise power and $I$ is the interference from other SBSs. Given a target transmission rate $r^a_{mi}$, the transmission power is thus
\begin{align}\label{shannon}
p^a_{mi}=(2^{\frac{r^a_{mi}}{W}}-1)(\sigma^2 + I)H_{mi}^{-1}
\end{align}
To simplify the notations, we assume that the expected data size of each task is a unit size. Therefore, the expected transmission delay and energy consumption of each task are $d^a_{mi} = 1/r^a_{mi}$ and $e^a_{mi} = p^a_{mi}/r^a_{mi}$, respectively. The MUE-SBS association cost of SBS $i$ is therefore
\begin{align}
C^{a}_i = &\sum_{m\in M_{ii}}\lambda^a_m(d^a_{mi}+\gamma e^a_{mi})\nonumber\\
&+\sum_{j\in\mathcal{N}}\sum_{m\in\mathcal{M}_{ij}}(1-\tau_m)\lambda^a_m (d^a_{mj}+\gamma e^a_{mj})\\
&+\sum_{j\in\mathcal{N}}\sum_{m\in\mathcal{M}_{ij}}\tau_m\lambda^a_m (d^a_{mi}+\gamma e^a_{mi}) \nonumber
\end{align}
where $\gamma$ is a normalization coefficient for delay cost and energy cost. The first term on the right-hand side is the association cost of authorized MUEs of SBS $i$ that only associate with SBS $i$. The second and third terms are the costs due to authorized MUEs of SBS $i$ associating with other SBSs $j\neq i$. However, for each MUE $m$, its private tasks must be sent to its own SBS $i$. Note that for a seller SBS $i$, we must have $\mathcal{M}_{ij} = \emptyset, \forall j$ since it has enough computing resources to accommodate all task requests from its own authorized MUEs.

\textbf{Stage II: SBS peer offloading}. In this stage, the cost of SBSs is mainly caused by task migration between SBSs due to transmission delay and energy consumption. Let $\beta_{ij}$ denote the amount of tasks sent from SBS $i$ to SBS $j$. The transmission rate from SBS $i$ to SBS $j$ and the associated transmission power are denoted by $r_{ij}$ and $p_{ij}$, respectively, which can be derived in a similar way as in Stage I. Therefore, the transmission delay and energy consumption for each task is $d^{s,tx}_{ij}=1/r_{ij}$ and  $e^{s,tx}_{ij}=p_{ij}/r_{ij}$. The transmission cost incurred to SBS $i$ due to peer offloading is therefore:
\begin{align}\label{cost_ij}
C^{s,tx}_{i}=\sum_{j\in\mathcal{N},j\neq i}\beta_{ij}\left(d^{tx}_{ij}+\gamma e^{tx}_{ij}\right)
\end{align}

\textbf{Stage III: SBS computing}. We then model the cost for processing computation tasks locally at SBSs. For each task, we assume that the required number of CPU cycles is an exponential random variable with mean $\rho$. The computational capability of SBS $i$ is measured by its CPU speed (i.e. CPU cycles per second), denoted by $f_i$. We model the computing delay using the M/M/1 queuing system. Thus the average computation delay (including task waiting time and processing time) for each task, can be obtained as
\begin{align}
d^{s,c}_i (\omega_i)= \frac{1}{f_i/\rho - \omega_i}
\end{align}
where $\omega_i$ is the workload processed at SBS $i$, which is the outcome of SBS coalition and load balancing and will be discussed shortly. The computation energy consumption for each task processed at SBS $i$ is proportional to the square of the CPU speed $(f_i)^2$, presented as $e^{s,c}_i= \kappa(f_i)^2$, where $\kappa$ is a constant depending on the CPU architecture \cite{mao2017mobile}. Therefore, the computation cost of SBS $i$ is
\begin{align}
	C^{s,c}_i=\omega_i(d^{s,c}_i(\omega_i)+\gamma e^{s,c}_i)
\end{align}
Notice that although here we use specific functions for computing delay and energy consumption, other functions can also be adopted. In practice, an SBS may maintain a lookup table for the expected delay and energy consumption under different workload inputs.

\textbf{Stage IV: SBS-to-Cloud offloading}. SBSs may still have to offload some computation tasks to the cloud. Let $\beta_{i0}$ be the number of computation tasks offloaded to the cloud by SBS $i$. Due to the large round-trip time to the remote cloud and the transmission energy consumption, the SBS-to-Cloud offloading cost of SBS $i$ is
\begin{align}\label{cost_i0}
	C^{c,tx}_{i}(\beta_{i0})=\beta_{i0}(d_{i0}^{c,tx}+\gamma e^{c,tx}_{i0})
\end{align}
where $d_{i0}^{c,tx}$ is the total expected delay (due to both transmission and computation on the cloud) from SBS $i$ to the cloud, and $e^{c,tx}_{i0}$ is the transmission energy consumed by SBS $i$ for each task.  In addition, the cloud charges SBSs $w_0$ \$/task for using the cloud service. Therefore, SBS $i$ will also incur an monetary cost
\begin{equation}
	M_i=w_0\beta_{i0}
\end{equation}

To sum up, the total operational cost of SBS $i$ is:
\begin{align}\label{offloading_cost}
	C_i=w_c(C^a_i+C^{s,tx}_i+C^{s,c}_i+C^{c,tx}_i) + M_i
\end{align}
where $w_c$ converts the delay and energy cost into a value comparable to the monetary cost. Keep in mind that this cost depends on how the SBSs form the coalition and balance the workload among themselves.

\subsection{SBS Social Trust Model}
Since SBSs are operated by individual owners, there are higher security and privacy risks for SBSs to offload their tasks to other SBSs than processing locally or offloading to the secure cloud. Therefore, when forming SBS coalitions, trust between SBSs must be taken into account.

Apart from the physical network of SBSs, we define a social trust network that describes the trust relationships between SBSs. Let $T_{ij} \in [0, 1]$ denote the trust value that SBS $i$ assigns to SBS $j$. $T_{ij} = 0$ indicates that SBS $i$ completely distrusts SBS $j$ and $T_{ij} = 1$ indicates that SBS $i$ completely trusts SBS $j$. However, even if two SBSs are neighbors in the physical network, they may not have an established trust relationship between each other.  For instance, a new SBS may have just been set up or two SBSs have not interacted with each other for a long time. As illustrated in Figure \ref{physical_social_network}, although SBS $i$ and SBS $j$ are physical neighbors and hence can potentially form a coalition, the values of $T_{ij}$ and $T_{ji}$ are unknown since they do not have a social relationship.

\begin{figure}[htb]
	\centering	
	\subfigure[Physical network]{\label{physical_network}
		\includegraphics[width=2 in]{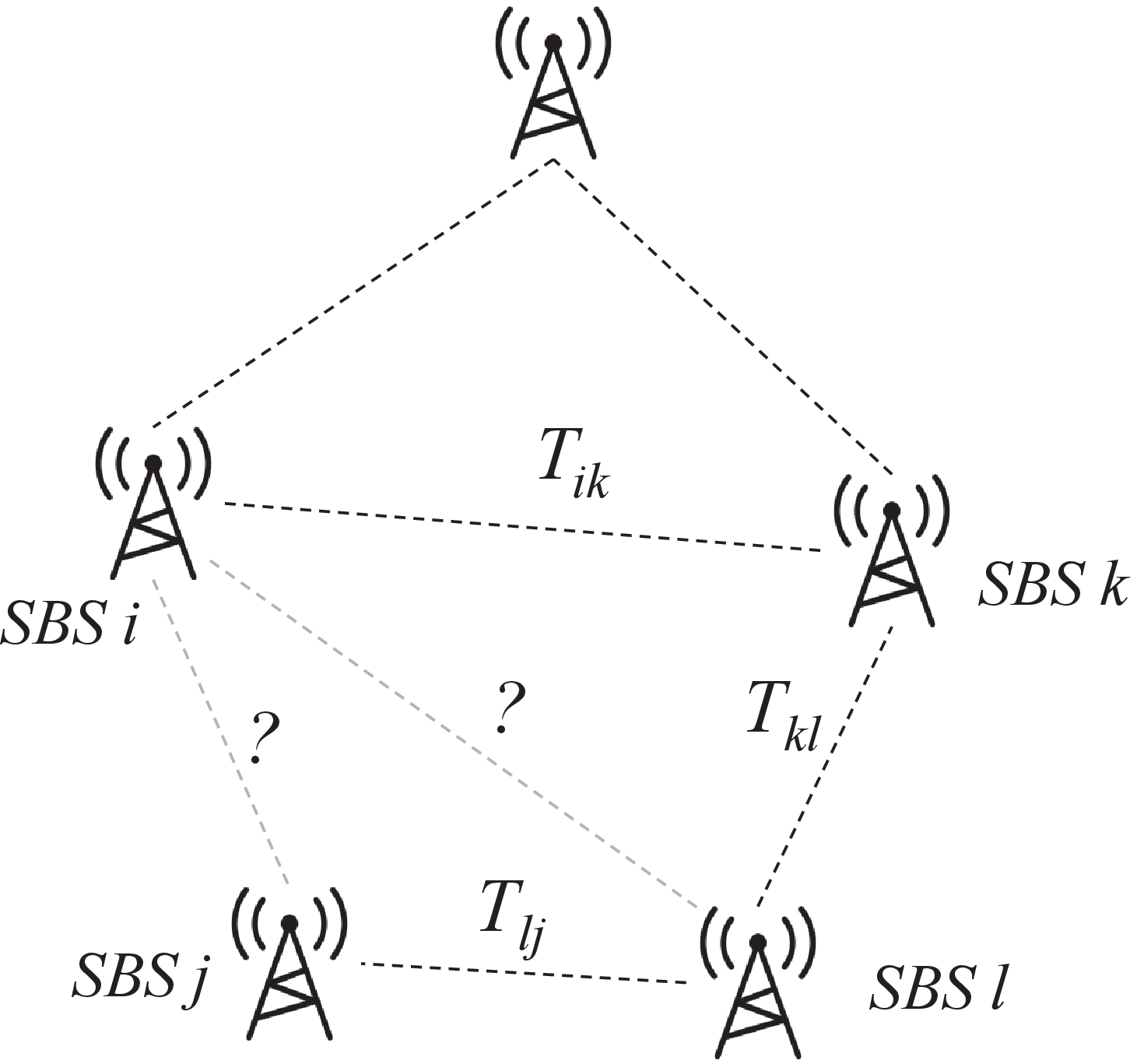}}
    \hspace{0.15in}
	\subfigure[Social trust network]{\label{social_network}
		\includegraphics[width=2 in]{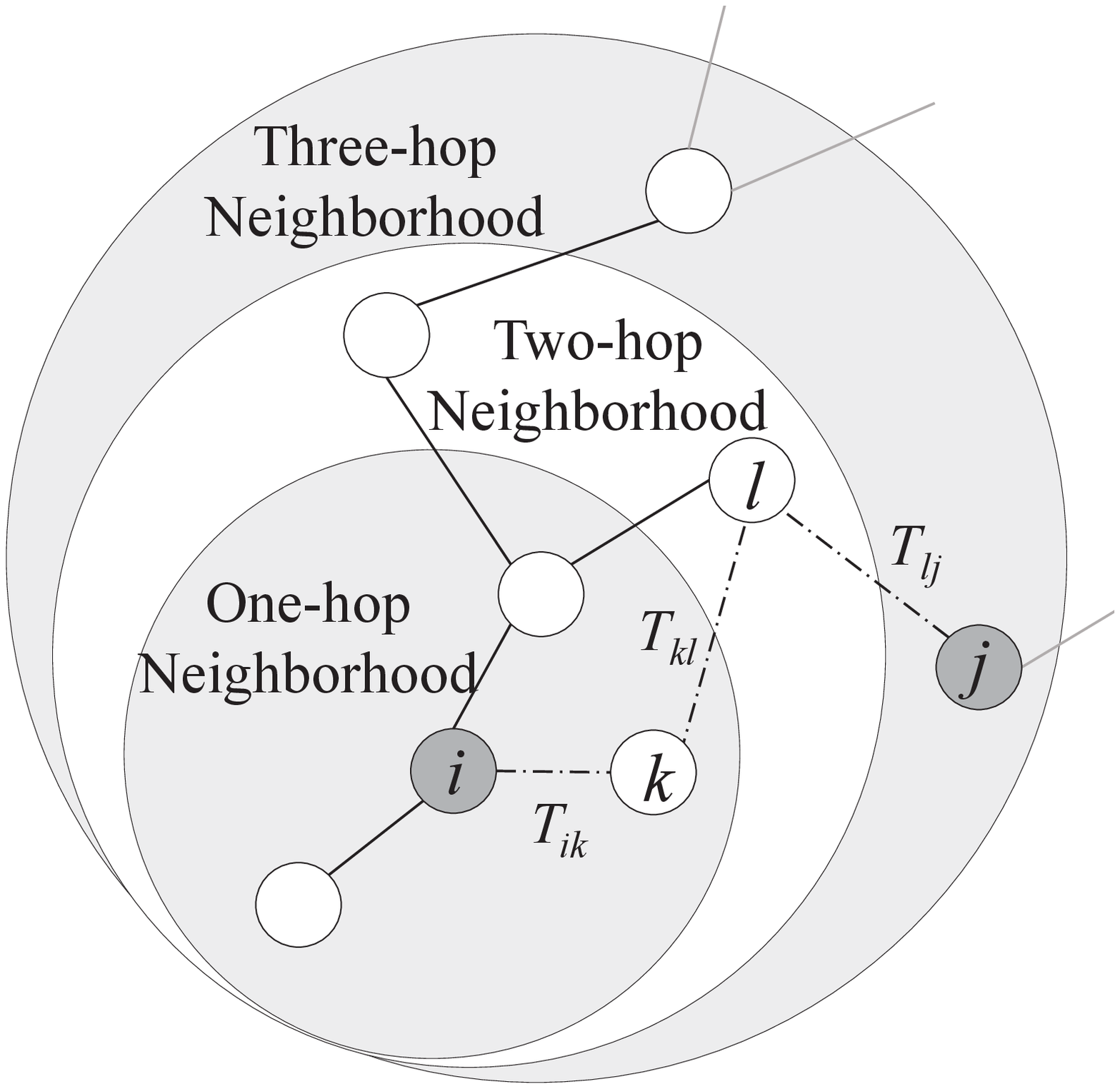}}
	\caption{Physical and social network.}
	\label{physical_social_network}
\end{figure}

In such cases, trust between the SBSs will be derived using the social trust network. If SBS $i$ can reach SBS $j$ via certain other SBSs in the social network, then the trust between SBS $i$ and SBS $j$ is computed by propagating trust along the path that connects them. Let $\Lambda(i,j)$ be the shortest path between SBS $i$ and SBS $j$, then trust $T_{ij}$ is computed as
\begin{align}
T_{ij} = \prod_{(k,l)\in \Lambda(i,j)}T_{kl}
\end{align}
We note that the above is just one common way to evaluate trust using the social trust network.  There are many other ways to evaluate trust proposed in the literature. Moreover, the trust values will be updated over time depending on the recent social interactions (including interactions in the coalitions) among SBSs.

Using the social trust network, the cost due to managing the security risk that SBS $i$ faces when offloading computation tasks to other SBSs (e.g. adopting stronger yet more costly security mechanisms) is modeled as follows:
\begin{equation}\label{offloading_risk}	R_i=w_r \sum_{j\in\mathcal{N}}(\beta^u_{ij}+\beta_{ij})(1-T_{ij})
\end{equation}
where $\beta^u_{ij}=\sum_{m\in\mathcal{M}_{ij}}(1-\tau_m)\lambda^a_m$ is the amount of offloaded workload due to direct MUE-SBS association, $\beta_{ij}$ is the amount of offloaded workload in SBS peer offloading and $w_r$ converts risk into a value comparable to the monetary cost.

\section{Collaborative SBS network as a Coalitional Game}\label{sec_coalition_game}
To formally study the formation of SBS coalitions and the resulting workload offloading decisions, we use the framework of the coalitional game theory \cite{bacsar1998dynamic,apt2009generic}. A coalitional game is defined as a tuple $(\mathcal{N},v)$ where $\mathcal{N}$ is the player's set and $v:2^{\mathcal{N}}\to \mathbb{R}$ is a function that assigns for every possible coalition $S\subseteq\mathcal{N}$ a real number representing the total benefit achieved by coalition $S$. By evaluating the values of different coalitions, players decide what coalitions are formed among themselves. In what follows, we first define the value function $v(S)$ for any given coalition $S\subseteq\mathcal{N}$. Clearly, $v(S)$ depends on not only which SBSs are in the coalition $S$ but also how they collaborate, namely how they perform load balancing. Then we will describe what coalitions are desired in terms of stability.
\subsection{Collaborative load balancing and value function for coalitions}
In this subsection, we investigate the interaction among SBSs that belong to any given coalition $S$ (which may not be stable though). Let $S_s\subset S$ denote the set of seller SBSs in $S$ and $S_b\subset S$ denote the set of buyer SBSs in $S$. We must have $S_s\cup S_b = S$. Although there are many approaches to match buyers with sellers in the coalition (e.g., double auction or other matching algorithms \cite{saad2009coalitional}), for ease of implementation, we consider a simple scheme similar to \cite{saad2011coalotional} that relies on the preference of the buyers inside the coalition to decide the workload offloading decisions among the SBSs. Specifically, buyers can act sequentially (according to some order $\Pi$) to acquire their needed computation resource as follows:

\begin{itemize}
\item[(1)] Buyer $b_i\in S_b$ requests to acquire its needed computation resource from seller $s_j\in S_s$ that will potentially yield the smallest offloading cost according to the SBS peer offloading scheme in Algorithm 1 (which returns results of $\mathcal{M}_{b_is_j}$ and $\beta_{b_is_j}$).
\begin{itemize}
	\item If seller $s_j$ can offer the required computation resource of buyer $b_i$, then the buyer does not act further.
	\item Otherwise, buyer $b_i$ acquires as much computation resource as possible from seller $s_j$  and then tries to fulfill the rest of its computation resource demand by acquiring resources from other sellers in $S_s$.
\end{itemize}
\item[(2)] Buyer $b_i$ repeats the above sequence until it has covered up all its computation resource deficit or no available sellers in $S_s$ exist. Then the next buyer SBS starts its acquiring process.
\end{itemize}

This process is repeated for all buyers in $S_b$. If a buyer is unable to find a seller in $S_s$ and still needs computation resource, then the buyer will offload the remaining computation tasks, represented by $\beta_{b_i 0}$ to the cloud. Essentially a buyer SBS tries to offload as much computation workload as possible to peer seller SBSs.

\begin{algorithm}[htb]
	\caption{SBS peer offloading scheme}
	\begin{algorithmic}[1]
		\Statex \textbf{Input}: Computation surplus $|\alpha_{s_j}|$ of seller $s_j$; Computation deficit $|\alpha_{b_i}|$ of buyer $b_i$.
		\State Maximum task offloading:  $\alpha=\min(|\alpha_{b_i}|, |\alpha_{s_j}|)$;
		\Statex \textbf{Stage 1}: \emph{MUE-SBS association}
		\State Find users co-covered by buyer $b_i$ and seller $s_j$, denoted by $\mathcal{R}_{b_is_j}$;
		\If {$\sum_{m\in\mathcal{R}_{b_is_j}}(1-\tau_m)\lambda^a_m\geq \alpha$}
		\State Choose $\mathcal{M}_{b_is_j}$ satisfying:
		\Statex $\qquad\qquad\sum_{m\in\mathcal{M}_{b_is_j}}(1-\tau_m)\lambda^a_m = \alpha$; \textbf{Stop};
		\Else
		\State $\mathcal{M}_{b_is_j}=\mathcal{R}_{b_is_j}$;
		\State $\tilde{\alpha}=\alpha-\sum_{m\in\mathcal{M}_{b_is_j}}(1-\tau_m)\lambda^a_m$; \textbf{Go to Stage 2};
		\EndIf
		\Statex \textbf{Stage 2}: \emph{SBS peer offloading}
		\State $\beta_{b_is_j}=\tilde{\alpha}$;
		\State \textbf{return} $\mathcal{M}_{b_is_j}$, $\beta_{b_is_j}$;
	\end{algorithmic}
\end{algorithm}

Following the above buyer-seller matching and workload allocation process, the values of $\mathcal{M}_{b_is_j}$, $\beta_{b_is_j}$ and $\beta_{b_i0}$ can thus be determined for each $b_i$ and $s_j$. In particular, the workload $\omega_{b_i}$ (or $\omega_{s_j}$) that buyer SBS $b_i$ (or seller SBS $s_j$) needs to process locally can be determined as follows:
\begin{align}
	\omega_{b_i}=\lambda^s_{b_i}-\sum_{s_j\in S_s}\sum_{m\in\mathcal{M}_{b_is_j}}(1-\tau_m)\lambda^a_m-\sum_{s_j\in \{S_s\cup\{0\}\}}\beta_{b_is_j}
\end{align}
\begin{align}
	\omega_{s_j}=\lambda^s_{s_j}+\sum_{b_i\in S_b}\sum_{m\in\mathcal{M}_{b_is_j}}(1-\tau_m)\lambda^a_m+\sum_{b_i\in S_b}\beta_{b_is_j}
\end{align}

With all these values derived, we are able to compute the operational cost $C_i$ as well as the risk management cost $R_i$ for each SBS $i \in S$ according to our system model. Clearly, the values of these costs also depend on the ordering of the buyer SBSs in the aforementioned matching process. Let $\mathcal{O}_S$ be the set of all possible orderings over buyers in $S$. Then given an ordering $\Pi\in\mathcal{O}_S$, the utility of SBS $i$ in coalition $S$ is thus defined as
\begin{align}\label{u_S_Pi}
	u_i(S, \Pi)=- (C_i + R_i), i\in S
\end{align}
and the total utility for coalition $S$ is
\begin{align}\label{u_S_Pi}
	U(S,\Pi)=\sum_{i\in S}u_i(S, \Pi)
\end{align}
The minus sign is inserted to turn the problem into a maximization problem in order to facilitate the analysis of the coalitional game. The value function for the SBSs coalition formation game $(\mathcal{N},v)$ is defined as
\begin{align}\label{value_function}
	v(S)=\max_{\Pi\in\mathcal{O}_S}U(S,\Pi)
\end{align}
which is the maximum achievable total utility over all possible orderings of the buyers, or equivalently, the minimum achievable total cost of SBSs in the coalition.

\subsection{Payment scheme within a coalition}
The previous subsection describes how SBSs in a coalition can collaborate with each other to improve their total utility (i.e. reduce their total cost). However, although the overall performance of the coalition may be improved, individual SBSs, especially the seller SBSs, do not share their computing resources with others for free. In this subsection, we design a payment scheme to provide seller SBSs with incentives to cooperate. In the next section, we will study what stable coalitions are formed under this payment-based incentive mechanism.

Let $g_i$ be the payment/reward of a buyer/seller SBS $i$, then its post-payment utility $\phi_i$ becomes
\begin{align}
\phi_i = u_i - g_i
\end{align}
Clearly, the total payment must  equal the total reward within a coalition and hence $\sum_{i \in S} g_i = 0$. If $g_i > 0$, then SBS $i$ pays $g_i$. If $g_i < 0$, then SBS $i$ receives $|g_i|$ reward.

Our payment scheme is developed by following a simple yet strict fairness criterion, namely \textit{proportional fairness payoff division}. Nevertheless, other fairness criteria, such as egalitarian fair, Shapley value, nucleolus, can also be adopted. In this scheme, the values of payments and rewards are decided by dividing the payoff (the utility improvement) of the whole coalition due to cooperation among the SBSs proportionally to their utility achieved without cooperation. Specifically, for SBS $i$, its post-payment utility will be
\begin{align}\label{get_phi}
	\phi_i=\psi_i\left(v(S)-\sum_{j\in S}v(\{j\})\right)+v(\{i\})
\end{align}
where $v(\{i\})$ represents the utility of SBS $i$ if it does not join any coalition (so only local processing or offloading to the cloud), and $\psi_i$ is the proportional weight satisfying $\sum_{i\in S}\psi_i=1$ and
\begin{equation}
	 \dfrac{\psi_i}{\psi_j}=\dfrac{\tilde{v}(\{i\})}{\tilde{v}(\{j\})}
\end{equation}
where $\tilde{v}(\{i\})$ is the normalized utility of SBS $i$ within its coalition. The normalization is to map negative SBS utilities to a positive interval. It is easy to verify that $\sum_{i \in S}\phi_i=v(S)$. Based on the proportional fairness criterion, the payment/reward of SBS $i$  can thus be determined as
\begin{align}
	g_i(S)=\phi_i(S)-u_i(S), \forall i\in S
\end{align}
In the above equation, $\phi_i$ can be interpreted as the expected utility of SBS $i$ by participating in the coalition, while the $u_i$ is the actually realized utility. The gap of the two is filled by the payment scheme.

There are two implementation issues for the payment scheme. First, payments need to be properly distributed since multiple buyers and sellers may be involved in the transaction. Moreover, direct payment from buyers to sellers faces fraud risks in the monetary transaction. To enable the effective and safe transaction, the edge orchestrator, which is a trusted third-party, collects payments from all buyers and then distributes them to the sellers.

\subsection{Stability of Coalitions}
SBSs may form multiple disjoint coalitions and there are many ways that SBSs form coalitions. However, we are interested in forming \textit{stable} coalitions such that no SBS or group of SBSs have incentives to leave the current coalition to form a different coalition. In particular, the requirement that all SBSs (buyers and sellers) in a coalition must at least receive higher utilities than working individually is a necessary but not sufficient condition for stability.

Consider any subset $\mathcal{K} \subseteq \mathcal{N}$ of SBSs, we call $\mathcal{S} = \{S_1, ..., S_L\}$ a collection of coalitions formed by these SBSs, where $S_l \subseteq \mathcal{K},\forall l$ are disjoint subsets of $\mathcal{K}$. If $\mathcal{K} = \mathcal{N}$, then we call $\mathcal{S}$ a partition of $\mathcal{N}$. A defection function $\mathbb{D}$ is a function that associates each possible partition $\mathcal{S}$ of $\mathcal{N}$ with a group of collections. The stability of a partition $\mathcal{S}$ is defined with respect to a defection function.
\begin{definition}($\mathbb{D}$-\textbf{stability}).
A partition $\mathcal{S}$ of $\mathcal{N}$ is $\mathbb{D}$-stable if no group of SBSs are interested in leaving $\mathcal{S}$ and forming a new collection of coalitions $\mathcal{S}'\in\mathbb{D}(\mathcal{S})$. That is, at least one SBS in such a group does not improve its utility by leaving the current partition.
\end{definition}
In other words, a defection function $\mathbb{D}$ restricts the possible ways that SBSs may deviate/defect. Two defection functions are of particular interest. The first function, denoted by $\mathbb{D}_c$, associates with each partition $\mathcal{S}$ the group of all possible collections in $\mathcal{N}$, namely there is no restriction on the way SBSs may deviate. The second function, denoted by $\mathbb{D}_{hp}$, associates each partition $\mathcal{S}$ with the group of collections that can be formed by merging or splitting coalitions in $\mathcal{S}$. Therefore, $\mathbb{D}_{hp}$-stability is weaker than $\mathbb{D}_{c}$-stability. In the next section, we design a distributed SBS coalition formation algorithm that achieves at least $\mathbb{D}_{hp}$-stability.

%


\section{Distributed Coalition Formation}\label{sec_CF}
\subsection{Distributed Coalition Formation Algorithm}
To present the distributed SBS coalition formation algorithm, we first introduce the notion of Pareto dominance to compare the ``quality'' of two collections of coalitions.
\begin{definition} (\textbf{Pareto Dominance})
Consider two collections of disjoint coalitions $\mathcal{S}_1$ and $\mathcal{S}_2$ formed by the same subset of SBSs $\mathcal{K} \subseteq \mathcal{N}$. $\mathcal{S}_1$ Pareto-dominates $\mathcal{S}_2$, denoted by $\mathcal{S}_1\rhd\mathcal{S}_2$, if and only if $\phi_i(\mathcal{S}_1) \geq \phi_i(\mathcal{S}_2), \forall i \in \mathcal{K}$ with at least one strict inequality for some SBS.
\end{definition}
Pareto dominance implies that a group of SBSs prefer to form coalitions in $\mathcal{S}_1$ rather than $\mathcal{S}_2$, if and only if at least one SBS is able to strictly improve its utility without hurting any other SBS. The following two operations, namely \emph{merge} and \emph{split} \cite{apt2009generic}, are central to our coalition formation algorithm:
\begin{itemize}
	\item \textbf{Merge}: merge a set of coalitions $\{S_1,\dots,S_l\}$ into a bigger coalition $\bigcup^{l}_{j=1}S_j$ if  $\{\bigcup^{l}_{j=1}S_j\}\rhd\{S_1,\dots,S_l\}$.
	\item \textbf{Split}: split a coalition $\{\bigcup^{l}_{j=1}S_j\}$ into a set of smaller coalitions $\{S_1,\dots,S_l\}$ if $\{S_1,\dots,S_l\}\rhd\{\bigcup^{l}_{j=1}S_j\}$.
\end{itemize}

By performing Merge, a group of SBSs can operate and form a single and larger coalition if this formation increases the utility of at least one SBS without decreasing the utility of any other involved SBSs. Hence, a Merge decision ensures that all involved SBSs agree on its occurrence. Likewise, a coalition can decide to split and divide itself into smaller coalitions if splitting is preferred in the Pareto sense.

\begin{algorithm}[htb]
	\caption{Distributed SBS coalition formation}
	\begin{algorithmic}[1]
		\State \textbf{Initial}: The SBS network is partitioned by $\mathcal{S}=\mathcal{N}=\{1,\dots,N\}$ with non-cooperative SBSs at the beginning of each operational time slot.
		\Statex \emph{Phase 1: SBS Coalition Formation}
		\State \textbf{Repeat}
		\State \quad (a) $\mathcal{S}^{\prime}\leftarrow$ Merge($\mathcal{S}$): SBS coalitions in $\mathcal{S}$ decide to merge by examining the Pareto dominance.
		\State \quad (b) $\mathcal{S}\leftarrow$ Split($\mathcal{S}^{\prime}$): SBS coalitions in $\mathcal{S}^{\prime}$ make distributed split decision using the Pareto dominance.
		\State \textbf{Until} Merge and split converges to a final partition $\mathcal{S}_f$
		\Statex \emph{Phase 2: Cooperative computation offloading}
		\State \quad (a) each coalition $S_i\in\mathcal{S}_f$ order its buyers in a way to minimize the offloading cost (maximize \eqref{value_function}).
		\State \textbf{Repeat} for every $S_i\in \mathcal{S}_f$
		\State \quad (b) each buyer in a coalition $S_i\in\mathcal{S}_f$ attempts to acquire computation demands in coalition $S_i$.
		\State \textbf{Until} no peer offloading in the coalition is available.
		\State \quad (c) any buyer who still has unsatisfied computation demand performs SBS-to-cloud offloading.
        \State These two stages are repeated periodically to adapt the partition to environmental changes.
	\end{algorithmic}
\end{algorithm}

Our SBS coalition formation algorithm is developed based on the Merge and Split operations, which is presented in Algorithm 2. The algorithm consists of two phases. The coalition formation phase iteratively executes the Merge and Split operations. Given the current partition $\mathcal{S}$, each coalition $S \in \mathcal{S}$ negotiates, in a pairwise manner, with neighboring SBSs to assess a potential merge. The two coalitions will then decide whether or not to merge. Whenever a Merge decision occurs, a coalition can subsequently investigate the possibility of a Split. Clearly, a Merge or a Split operation is a distributed decision that an SBS (or a coalition of SBSs) can make. After successive Merge-and-Split iterations, the network converges to a partition composed of disjoint coalitions and no coalition has any incentive to further merge or split. In other words, the partition is \emph{Merge-and-Split proof}. The convergence of any Merge-and-Split iterations such as the proposed algorithm is guaranteed as shown in \cite{apt2009generic}. Upon convergence, the second phase of the actual computation offloading then starts using mechanisms described in Section \ref{sec_coalition_game}.

\subsection{Stability Analysis}
The outcome of the above algorithm is a partition of disjoint independent coalitions of SBSs. As an immediate result of the definition of $\mathbb{D}_{hp}$ stability, every partition resulting from proposed algorithm is $\mathbb{D}_{hp}$-stable. In particular, no coalitions of SBSs in the final partition have the incentive to pursue a different coalition formation through Merge or Split. Next, we investigate whether the proposed algorithm can achieve $\mathbb{D}_{c}$-stability.

A $\mathbb{D}_c$-stable partition has the following properties according to \cite{apt2009generic}. (i) No SBSs are interested in leaving $\mathcal{S}$ to form other collections in $\mathcal{N}$ (through any operation). (ii) A $\mathbb{D}_c$-stable partition is the \emph{unique} outcome of any \emph{arbitrary} iteration of merge-and-split, if it exists. (iii) A $\mathbb{D}_c$-stable partition $\mathcal{S}$ is a unique $\rhd$-maximal partition, i.e., for all partition $\mathcal{S}'\neq\mathcal{S}$, we have $\mathcal{S}\rhd\mathcal{S}^\prime$. Therefore, the $\mathbb{D}_c$-stable partition provides a \emph{Pareto optimal} utility distribution. However, the existence of a $\mathbb{D}_c$ stable partition is not always guaranteed \cite{apt2009generic}. Nevertheless, we can still have the following result.
\begin{theorem}
The proposed distributed SBS coalition formation algorithm converges to the Pareto-optimal $\mathbb{D}_c$-stable partition, if such a partition exists. Otherwise, the final partition is $\mathbb{D}_{hp}$-stable.
\end{theorem}
\begin{proof}
The proof is immediate due to the fact that, when it exist, the $\mathbb{D}_c$-stable partition is a unique outcome of any Merge-and-Split iteration \cite{apt2009generic}, such as any partition resulting from our coalition formation algorithm.
\end{proof}

The stability of the grand coalition (e.g. all SBSs form a single coalition) is of particular interest in the coalitional game theory. It can be easily shown that the considered SBS coalitional game is generally not superadditive and its core is generally empty due to the extra offloading cost (transmission delay, cost and collaboration risk) and hence, the grand coalition is not stable. Instead, independent disjoint coalitions will form. Readers who are interested in more details on the stability of grand coalition in coalitional games are referred to \cite{saad2009coalitional,bogomolnaia2002stability}.

\subsection{Complexity Analysis}
The complexity of the proposed coalition formation algorithm lies mainly in the complexity of the Merge and Split operations. For a given network, in one Merge operation, each current coalition attempts to merge with other coalitions in a pairwise manner. In the worst case scenario, every SBS, before finding a suitable merge partner, needs to make a merge attempt with all other SBSs in $\mathcal{N}$. In this case, the first SBS requires $N-1$ attempts for merge, the second requires $N-2$ attempts and so on. The total number of merge attempts in the worst case is thus $\sum^{N}_{i=1}(N-i)$. In practice, the merge process requires a significantly lower number of attempts since finding a suitable partner does not always require to go through all possible merge attempts (once a suitable partner is identified the merge will occur immediately). The complexity is further reduced due to the fact that SBSs do not need to attempt to merge with physically unreachable SBSs. Moreover, after the first run of the algorithm, the initial $N$ non-cooperative SBSs will self-organize into larger coalitions. Subsequent runs of the algorithm will deal with a network composed of a number of coalitions that is much smaller than $N$.

For the split operation, in the worst case scenario, splitting can involve finding all the possible partitions of the set formed by the SBSs in a single coalition. For a given coalition $S$, this number is given by the Bell number $\sum^{|S|}_{k=1} \binom{|S|}{k}$ which grows exponentially with the number of SBSs $|S|$ in the coalition. In practice, this split operation is restricted to the formed coalitions, and thus it will be applied to small sets. The split complexity is further reduced due to the fact that, in most scenarios, a coalition does not need to search for all possible split forms. For instance, once a coalition identifies a suitable split structure, the SBSs in this coalition will split, and the search for further split forms is not needed in the current iteration.

\section{Simulation}\label{sec_simulation}
In this section, we conduct systematic simulations in practical scenarios to evaluate the performance of the proposed system and algorithm.
\subsection{Setup}
Our simulation adopts the widely-used stochastic geometry approach for ultra dense SBS deployment, which is modeled as a homogeneous Poisson Point Process (PPP) \cite{baccelli2010stochastic}. Specifically, we simulate a 100m$\times$200m$\times$50m office building where a set of SBSs are deployed whose locations are chosen according to the PPP with density 0.15. The distribution of MUEs also follows another PPP with density 0.6. The task generation process of each MUE is modeled as a Poisson process with a rate of 5 tasks per time slot. The maximum transmission power of MUEs is set as 10 dBm. The channel model follows the ITU indoor path loss model \cite{series2012propagation}: $L\text{[db]}=20\lg f+10\nu\lg d_u+L_f(n)-28$ (The values and corresponding explanations of parameters are shown in TABLE \ref{para_set}). The target MUE-SBS transmission rate is $r_u=25$ Mbps and hence the corresponding transmission power can be computed. Similar requirements are imposed on the transmissions between SBSs with SBS maximum transmission power 20 dBm and target transmission rate $r_s=50$ Mbps. Figure \ref{SBS_deployment} shows the SBS deployment and MUE association in one operational time slot used in the simulation.
\begin{table}
	\renewcommand\arraystretch{1}
	\centering
	\caption{Simulation setup: system parameters}
	\begin{tabular}{l|c}
		\hline
		
		Parameters & Value\\
		\hline
		\hline
		Maximum UE transmission power $p_m$   & 10 dBm \\
		Maximum SBS transmission power $p_s$  & 20 dBm \\
		UE transmission rate requirement $r_u$ & 25 Mbps\\
		SBS peer transmission rate requirement $r_s$ & 50 Mbps\\
		Frequency (indoor path loss model) $f$ & 900 MHz \\
		Indoor path loss exponent $\nu$ & 3.3\\
		Floor penetration loss $L_f(n)$, $n$=[1,2,3]& [9,12,24] db\\
		Noise power $\sigma^2$ & -126.2 db\\
		SBS density (PPP)   & 0.15\\
		User density (PPP) & 0.6 \\
		Task arrival rate of MUEs $\lambda^a$ & 5\\
		System bandwidth $W$ & 20 MHz\\
		Cloud offloading delay $d^{c,tx}$ & 0.3 sec/task\\
		\hline
	\end{tabular}
\label{para_set}
\end{table}

\begin{figure}[htb]
	\centering	
	\includegraphics[width=3.5 in]{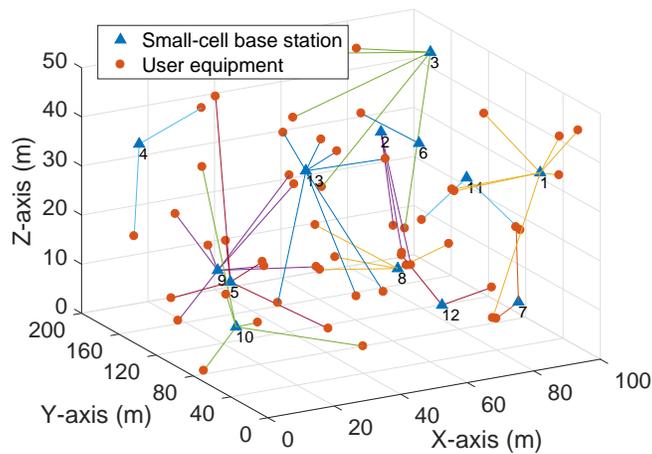}
	\caption{SBS deployment and MUE association}
	\label{SBS_deployment}
\end{figure}

For simplicity, we consider a fixed social trust network throughout the simulation time horizon. Nevertheless, our algorithm is also compatible with evolving trust networks that are updated continuously or periodically according to SBS interactions. The adopted social trust network is illustrated in Figure \ref{social_net}. SBSs that are physically unreachable will not have mutual trust values (e.g. SBS 3 and SBS 10).  SBSs that are physical neighbors may also miss mutual trust due to the lack of recent interactions. In this case, a trust value is obtained by finding the shortest path in the social trust network (e.g. SBS 5 $\rightarrow$ SBS 10 $\rightarrow$ SBS 13).
\begin{figure}[htb]
	\centering	
	\includegraphics[width=3.5 in]{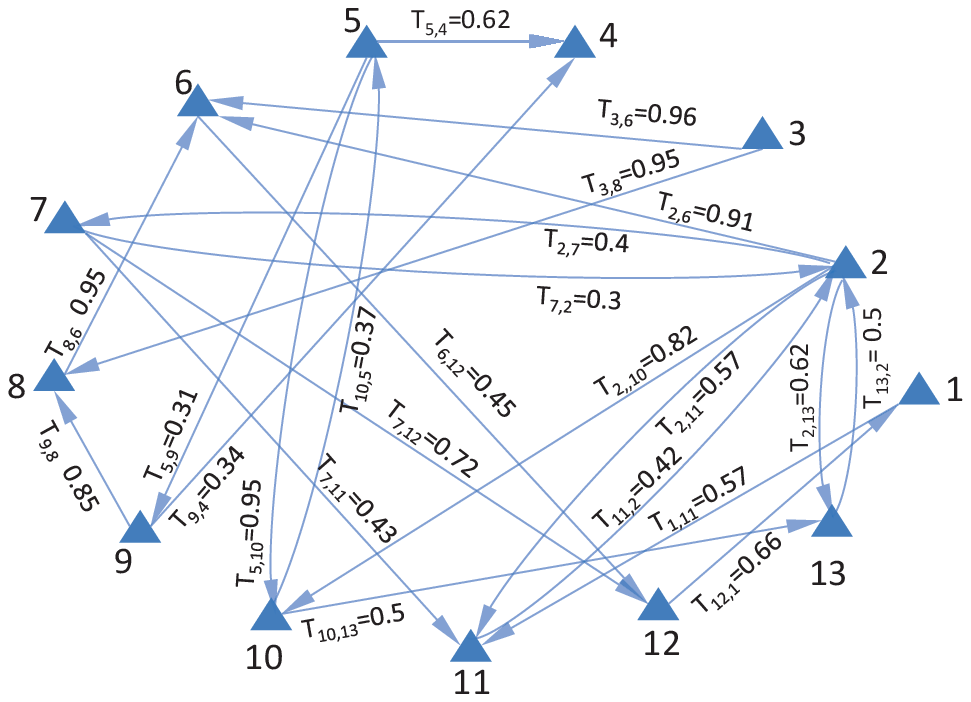}
	\caption{Social trust network}
	\label{social_net}
\end{figure}

The proposed SBS collaboration scheme is compared with three benchmark schemes:
\begin{itemize}
	\item \textbf{Non-cooperative scheme}: Every SBS does not share computation resources with peer SBSs. For overloaded SBSs, all unsatisfied computation tasks will be offloaded to the cloud server.
	\item \textbf{Cloud-offloading minimization}: The scheme greedily exploits the computation resources of peer SBSs to minimize the number of tasks offloaded to the cloud server. Therefore, peer offloading is performed whenever possible regardless of the offloading and risk management costs.
	\item \textbf{Centralized collaborative scheme}: This scheme uses brutal force to search for the best coalition structure that minimizes the total system cost without considering the coalition stability or SBS incentives.
\end{itemize}

\subsection{Coalition formation in one time slot}

\begin{figure}[htb]
	\centering	
	\includegraphics[width=3.2 in]{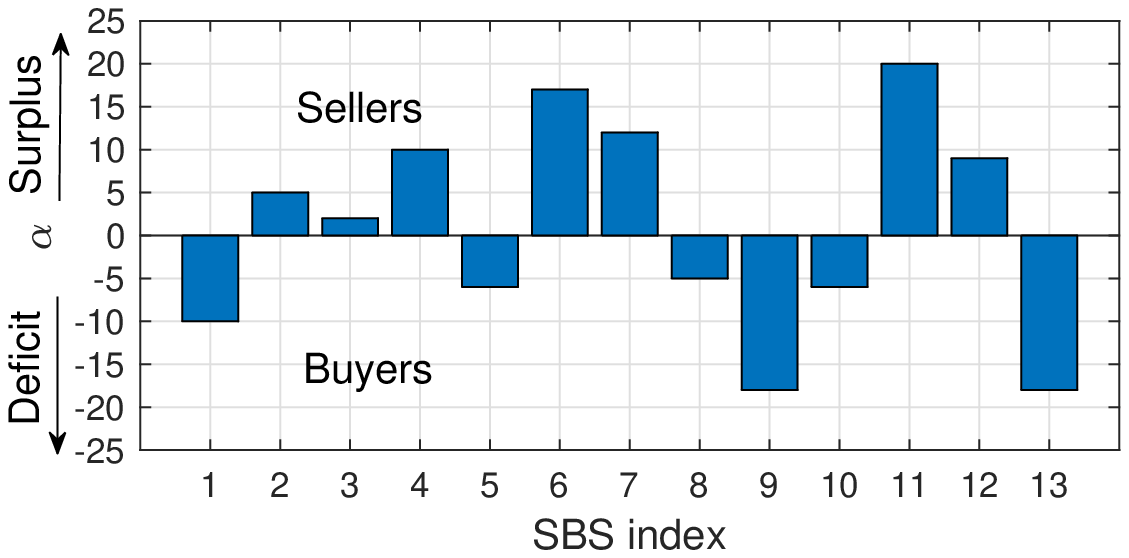}
	\caption{Computation resource surplus/deficit}
	\label{alpha}
\end{figure}


We exemplify the distributed coalition formation process by considering the computation surplus/deficit profile shown in Figure \ref{alpha} in one operational time slot. Figure \ref{phi_evo} depicts the evolution of the intermediate post-payment utilities $\bm{\phi}$ during the coalition formation process, where weights in the cost function are set as $w_c=0.2, w_r=0.2, w_0=1$. The horizontal axis represents the iterations in the Merge-and-Split process and for each iteration, we indicate if it is a Merge operation or a Split operation. Since all SBSs have their own computation tasks to process, they incur positive costs (hence negative utilities). In particular, SBSs 1, 5, 8, 9, 10, 13 have computation resources deficits. Without SBS collaboration, they have to offload the extra computation tasks to the remote cloud, thereby incurring large transmission delay and cloud payment costs. During the coalition formation, each iteration is executed by following the Merge/Split operation that aims to find a Pareto-dominant coalition partition than the current partition. Therefore, after each iteration, at least one of the SBSs improves its utilities without decreasing the utilities of other SBSs. Figure \ref{total_utility} shows the system utility evolution of each specific Merge or Split operation. We see that the system utility is improved with every Merge/Split operation and after only several iterations, the network converges to a stable partition of coalitions. This indicates that in practice, the complexity of running the proposed algorithm is low and hence, it can be easily implemented.

\begin{figure}[htb]
	\centering	
	\includegraphics[width=3.5 in]{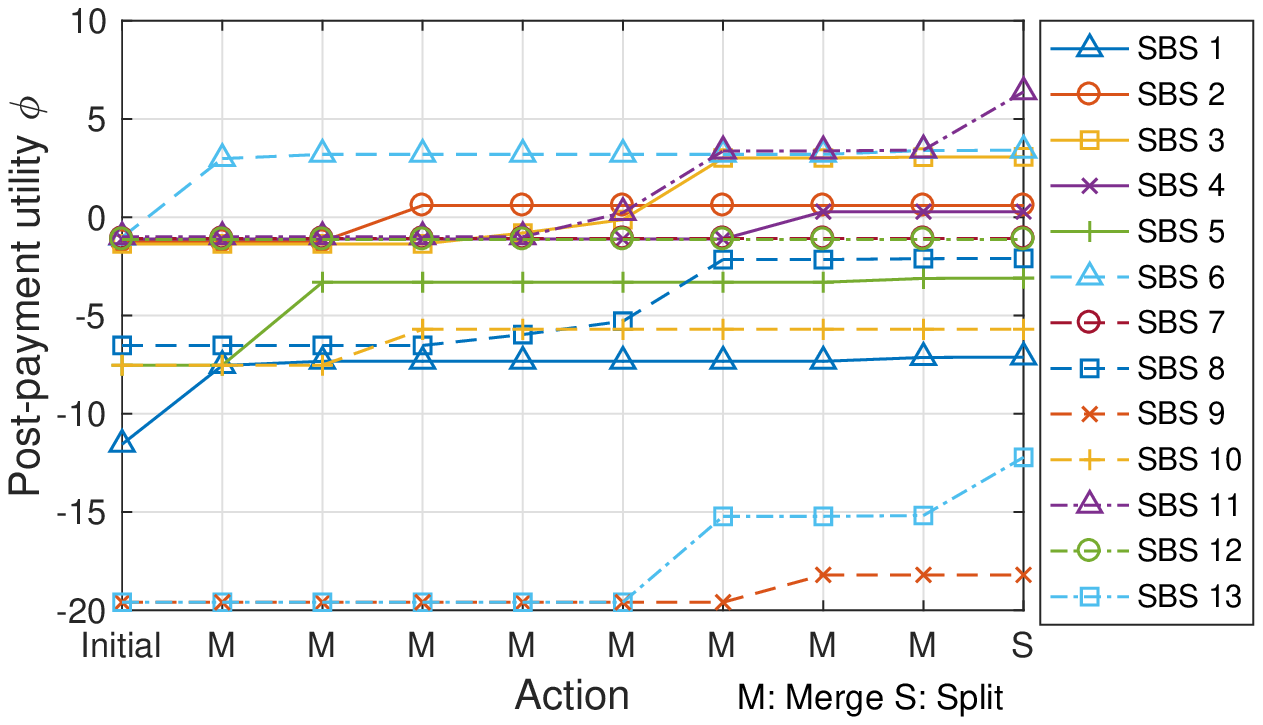}
	\caption{Evolution of post-payment utility $\bm{\phi}$}
	\label{phi_evo}
\end{figure}

\begin{figure}[htb]
	\centering	
	\includegraphics[width=3.5 in]{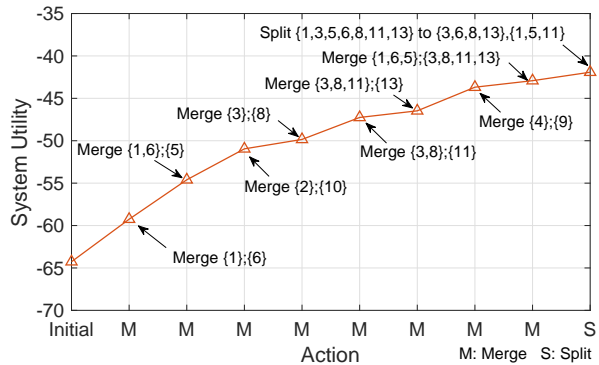}
	\caption{System utility evolution and coalition formation}
	\label{total_utility}
\end{figure}

The final coalitions and the number of exchanged computation tasks are presented in Figure \ref{coalition_offloadings}, where buyers are matched with nearby sellers. Several points are worth noting. First, a coalition can contain multiple buyers and multiple sellers and is not necessarily just a matching between one buyer and one seller. For example, the coalition $\{4, 9\}$ involves only one buyer SBS 9 and seller SBS 4, whereas the coalition $\{3, 6, 8, 13\}$ involves two buyer SBSs 8, 13 and two seller SBSs 3, 6. In particular, SBS 6 is sharing its computation resources with both SBSs 8 and 13. Second, an SBS may not want to join any coalition. In other words, an SBS may want to form an isolated coalition that contains only itself. In this particular simulation, we observe that SBS 7 and SBS 12 separately form isolated coalitions. As a result, the utilities of these two SBSs stay the same before and after the coalition formation. This can also be observed in Figure \ref{phi_evo}.

By reading both Figure \ref{alpha} and Figure \ref{coalition_offloadings}, we can see that not all computation demands of buyers can be satisfied via SBS coalition. For instance, consider SBS 9, it has a computation resource deficit for 18 tasks, yet it can only offload 10 tasks to SBS 4, which is the only matched seller SBS. In this case, the remaining 8 unsatisfied tasks will be offloaded to the cloud. Figure \ref{payment_reward} further shows the computation resource deficit (or surplus) and the actual computation resource bought (or sold) for each SBS. Clearly, the deficit (or surplus) serves an upper bound on what is actually bought (or sold) and hence, the magnitude of the red bar is always larger than that of the corresponding blue bar. In particular, SBSs 7 and 12 do not sell any of their computation resource surplus to other SBSs. In the same figure, we also show the payments made by buyer SBSs and the rewards received by seller SBSs. It can be verified that the total payment equals the total reward within each coalition. For instance, in the coalition $\{3, 6, 8, 13\}$, the total payment is 0.19 + 8.30 made by buyer SBSs 8, 13 and the total reward is 4.13 + 4. 36 received by seller SBSs 3, 6.

\begin{figure}[htb]
	\centering	
	\includegraphics[width=3.5 in]{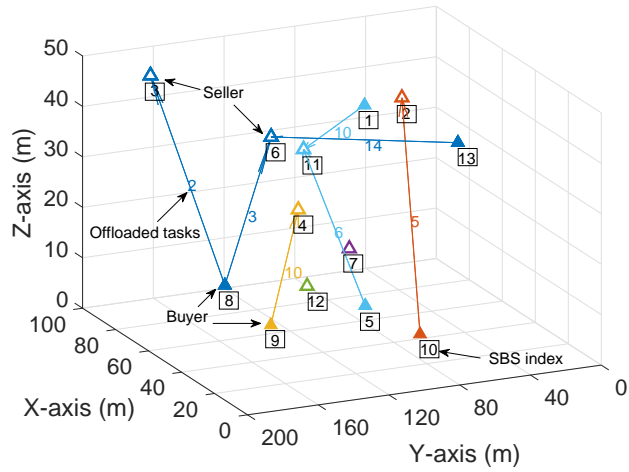}
	\caption{Coalitions and computation peer offloading}
	\label{coalition_offloadings}
\end{figure}

\begin{figure}[htb]
	\centering	
	\includegraphics[width=3.5 in]{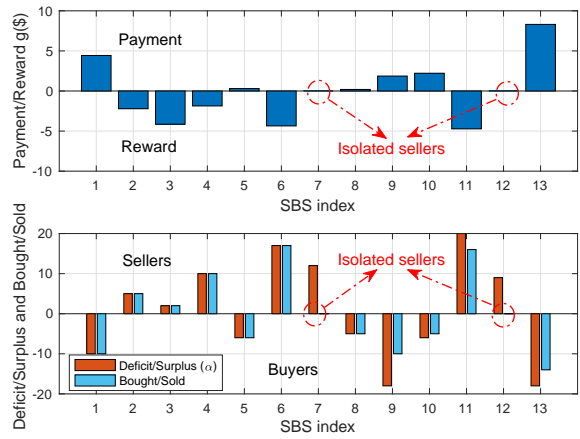}
	\caption{Payment/reward and bought/sold of SBSs}
	\label{payment_reward}
\end{figure}

We note that what coalitions will be formed depends on the adopted cost function. In this set of simulations, we investigate the impact of the relative weights of offloading cost $w_c$, risk management cost $w_r$ and payment to the cloud $w_0$. Table \ref{weight_and_coalition} reports the formed SBS coalitions for different weight profiles. We can see that, although in our scenario the grand coalition is difficult to form due to the physical network constraints, larger coalitions tend to form with smaller weights assigned to the peer offloading cost and the risk management cost. Figure \ref{impact_of_weight} presents the number of formed coalitions and the average coalition size.

\begin{table}
	\renewcommand\arraystretch{1}
	\centering
	\caption{Weights and coalition structure}
\begin{tabular}{ll}
		\hline
		[$w_c, w_r, w_0$] & SBS Coalition Structure\\
		\hline
		$[0.1,0.1,1.0]$ & $\left\{[1,2,10,11];[3,4,5,6,8,9];[7,12,13]\right\}$ \\
		$[0.1,1.0,1.0]$ & $\left\{[1,4,6,7];[2,5,8];[7,13,10,11];[3];[12]\right\}$ \\
		$[1.5,0.1,1.0]$ & $\left\{[1,11];[4,5];[6,8,9];[2];[3];[7];[10];[12];[13]\right\}$\\
		$[1.5,0.5,1.0]$ & $\{[1,11];[4,5];[6,9];\&~\text{the rest isolated SBSs}\}$\\
	 	$[1.0,1.0,1.0]$ & $\{[1,6,8];[4,5];[11,13];\&~\text{the rest isolated SBSs}\}$\\
		$[1.5,1.0,1.0]$ & $\{[1,11];[6,9];\&~\text{the rest isolated SBSs}\}$\\
		\hline
	\end{tabular}
\label{weight_and_coalition}
\end{table}

\begin{figure}[htb]
	\centering	
	\hspace{-0.1 in}
	\subfigure[Number of coalitions]{\label{num_coalition}
		\includegraphics[width=2.7 in]{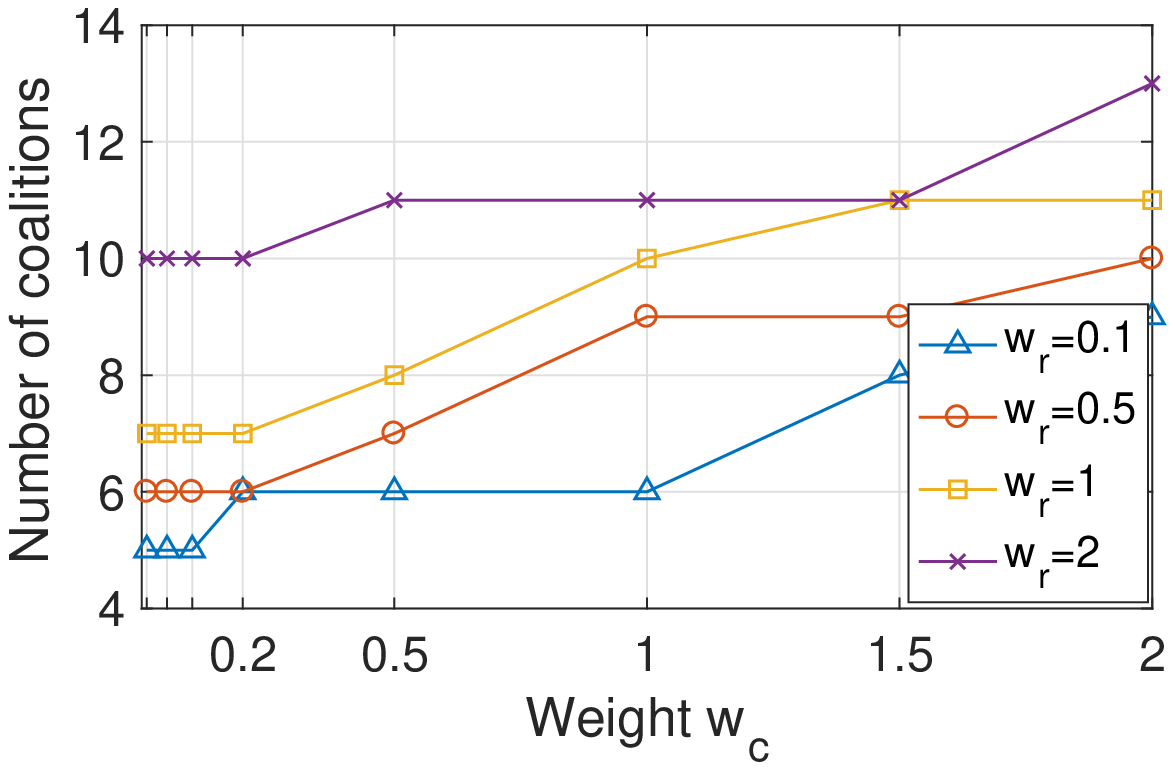}}
	\hspace{0.15 in}
	\subfigure[Average coalition size]{\label{max_coalition_size}
		\includegraphics[width=2.7 in]{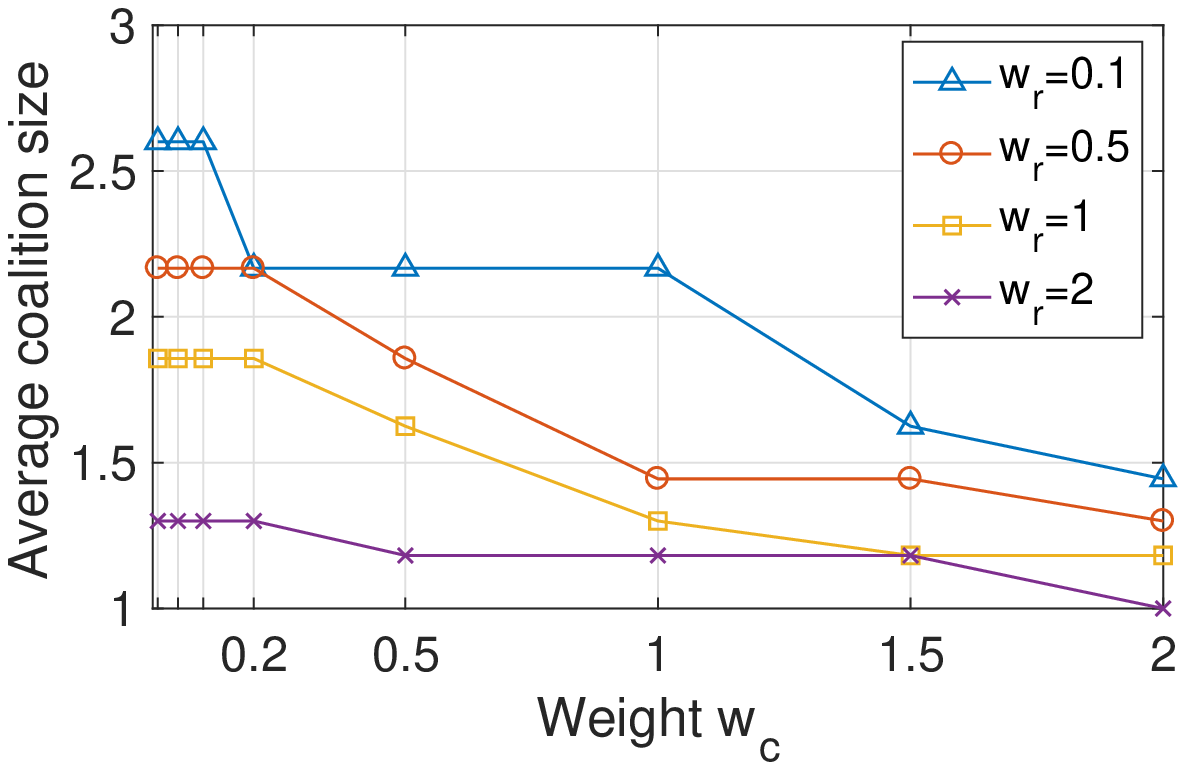}}
	\caption{Impact of weights on coalition formation ($w_0=1$).}
	\label{impact_of_weight}
\end{figure}

\subsection{Performance evaluation and comparison}
Figure \ref{costs_CF} shows the cost incurred by different parts of the system and offers a comparison with the three benchmarks. (1) The non-cooperative scheme incurs the highest total cost. Since there is no collaboration among the SBSs, each SBS has to offload its extra computation tasks to the remote cloud, thereby incurring a large transmission delay and cloud payment cost. However, since there is no task offloading between SBSs, its operational cost is lowered and totally avoids the risk management cost since all tasks are processed locally or in the secure cloud. (2) The purpose of the cloud-offloading minimization scheme is to minimize the number of tasks offloaded to the remote cloud, thereby reducing the transmission delay and cloud payment cost. As can be seen, the cost due to using cloud service is significantly reduced. However, since the optimization ignores the operation cost due to offloading among SBSs and especially the risk management cost, the total cost is still high. (3) The centralized collaborative scheme minimizes the overall network cost by jointly considering all cost sources, thereby achieving a much lower total cost. However, the centralized scheme is performed under the assumption that all SBSs are cooperative with no incentive issues, which is not suitable for our considered problem where SBSs are self-interested. (4) Our proposed solution based on the coalitional game gracefully addresses the incentive challenge. As can be seen, the proposed algorithm achieves similar performance to the centralized scheme while ensuring that the adopted coalitions are stable.

\begin{figure}[tb]
	\centering	
	\includegraphics[width=3.5 in]{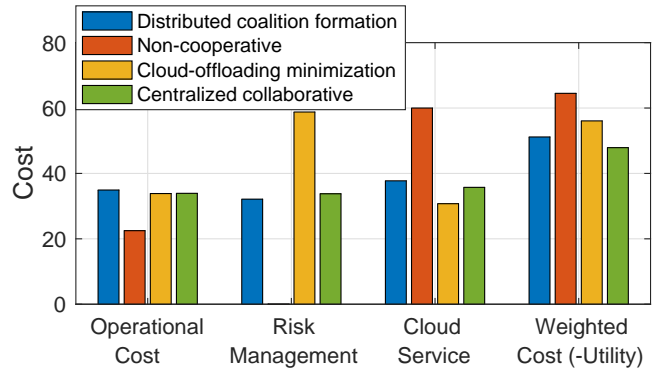}
	\caption{Costs characterization}
	\label{costs_CF}
\end{figure}

\begin{figure}[tb]
	\centering	
	\includegraphics[width=3.5 in]{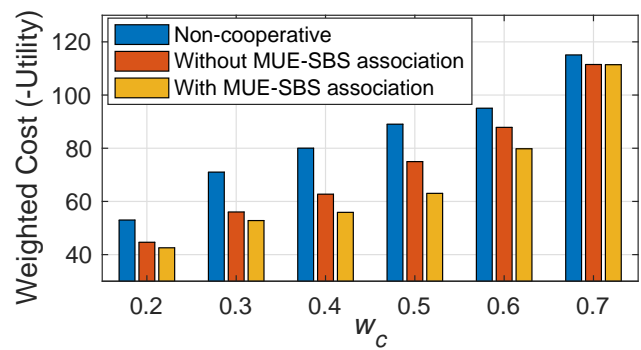}
	\caption{Advantage of MUE-SBS association scheme}
	\label{mue_assoc}
\end{figure}

\begin{figure}[tb]
	\centering	
	\includegraphics[width=3.5 in]{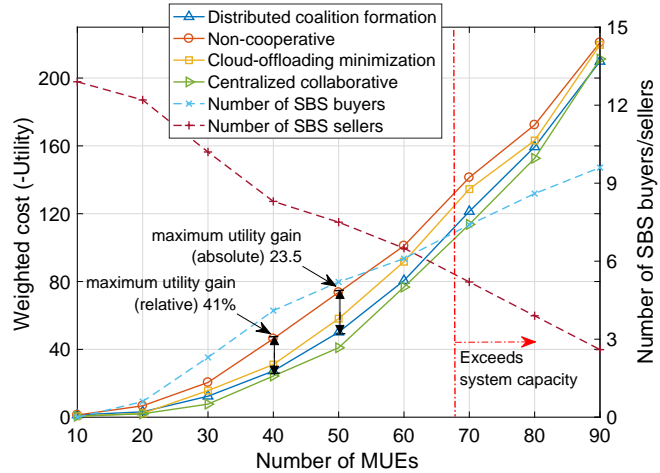}
	\caption{Impact of MUE number}
	\label{user_num}
\end{figure}

\begin{figure}[tb]
	\centering	
      \includegraphics[width=3.5 in]{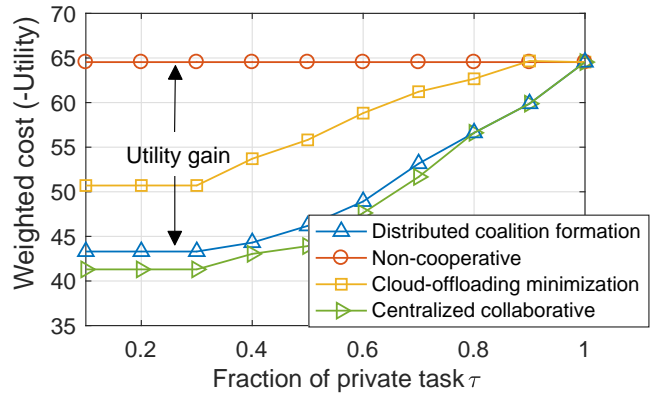}
	  \caption{Impact of private fraction $\tau$}
	  \label{private_fraction}
\end{figure}

\begin{figure*}[htb]
    \begin{minipage}[t]{0.5\linewidth}
	   \includegraphics[width=3 in]{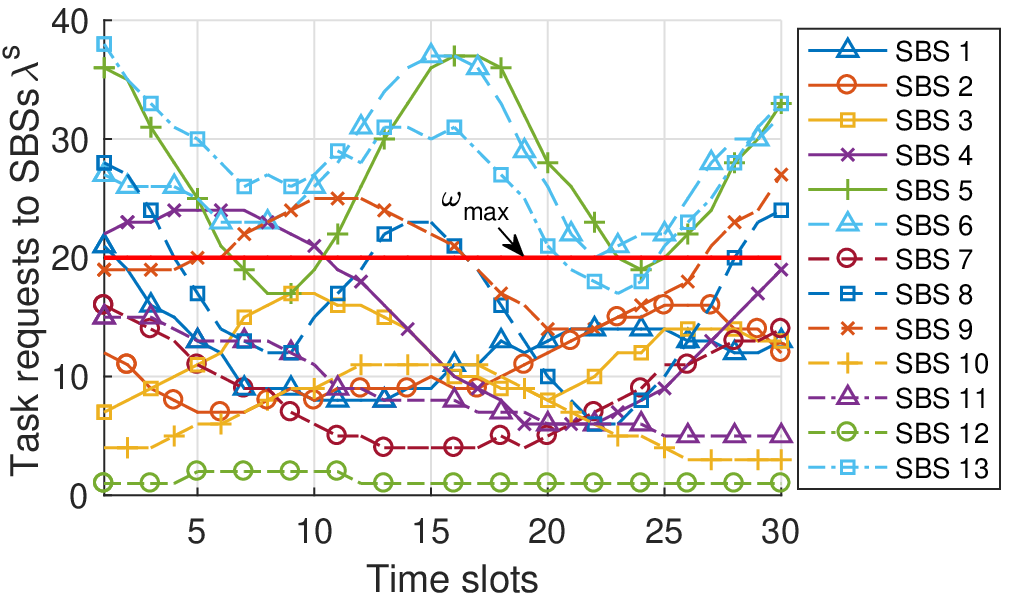}
	   \caption{Dynamics of task arrival}
	   \label{dynamic_lambda_s}
     \end{minipage}%
    \begin{minipage}[t]{0.5\linewidth}
	   \includegraphics[width=3 in]{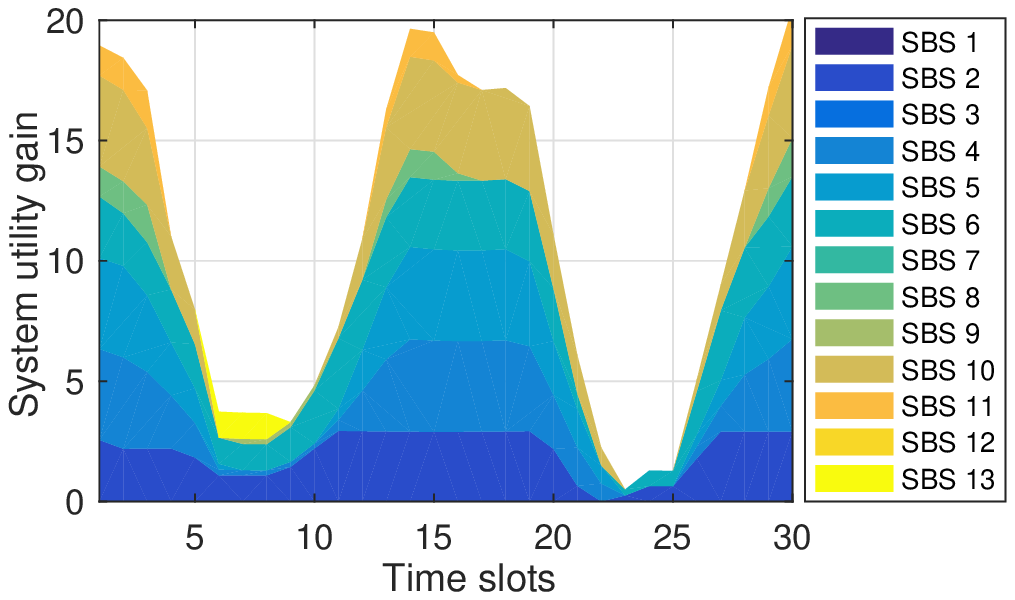}
	   \caption{Dynamics of utility gain}
	   \label{dynamic_utility_gain}
    \end{minipage}%
\end{figure*}

In Figure \ref{mue_assoc}, we demonstrate the benefits achieved by introducing the workload balancing in MUE-SBS association, which is a distinct feature of ultra dense SBS networks and that existing work \cite{tanzil2016distributed} did not consider. We run the coalition formation algorithm under two settings: with and without MUE-SBS association scheme. As can be observed, the MUE-SBS association scheme helps to reduce the system cost and the cost reduction depends on the value of $w_c$. When $w_c$ is small, the transmission delay and energy consumption saved by MUE-SBS association scheme are less valued, therefore, a modest cost reduction is realized. As the $w_c$ increases, the edge system benefits more by adopting the MUE-SBS association scheme. However, as shown previously, a too large $w_c$ discourages formation of SBS coalitions which restrains the role of MUE-SBS association. Hence, the benefit of MUE-SBS association scheme diminishes when the $w_c$ grows larger than 0.6.

How much performance improvement (i.e. cost reduction) can be achieved by SBS coalition will depend on the spatial traffic pattern in the network. Intuitively, if all SBSs have a light workload, then there is no need for collaboration, whereas if all SBSs are overloaded, then it is not possible to collaborate. Now, we investigate the performance of collaborative edge computing as a function of the system utilization level (i.e. ratio of expected task number in the network to system computation capacity $\sum_{i\in\mathcal{N}}\omega_i^{\max}$). We simulate a spectrum of the system utilization level by changing the number of MUEs in the network. Figure \ref{user_num} depicts the impact of the number of MUEs on the collaborative performance in terms of total system cost. When there are only a few MUEs, most of the SBSs can use its own computation resource to satisfy the computation requests. The system is thus a ``buyers' market''. When there is an excessive number of MUEs, most of the SBSs need extra computation resources. The system is thus a ``sellers' market''. In both cases, coalitions are difficult to form. The maximum utility gain by collaboration occurs when the system computation capacity matches the number of MUEs. By allowing collaboration among SBSs, the spatial workload intensity heterogeneity is mitigated via workload balancing. In our simulations, the maximum absolute utility gain is achieved at around 50 MUEs and the maximum relative utility gain (41\%) is achieved at around 40 MUEs.

The performance improvement by SBS coalition also depends on how flexibly computation tasks can be offloaded. Recall that computation tasks can be either normal tasks or private tasks that must be processed locally or offloaded to the remote cloud. Figure \ref{private_fraction} shows the impact of the fraction of the private tasks among the total tasks. As can be seen, a larger utility gain is achieved at a lower fraction of private tasks because peer offloading among SBSs is more flexible. Again the proposed coalition formation algorithm significantly reduces the total system cost and achieves close-to-optimal performance.

In order to evaluate the performance of proposed algorithm over multiple operational slots, we further simulate a dynamic edge system with random MUE arrival and computation task arrival; the social trust network is assumed to be observed and updated every 5 slots. The total task requests received by SBSs are given in Figure \ref{dynamic_lambda_s}. With the temporally heterogeneous computation task requests, the SBSs intermittently switch between ``buyer'' and ``seller'' modes. Figure \ref{dynamic_utility_gain} shows the corresponding system utility gain across 30 operational slots, where each color block denotes the post-payment utility of an SBS. It can be observed that the system utility gain is mainly decided by the buyer-seller composition in the edge system. For example in time slot 23, there is only one ``buyer'' (SBS 6) in the system and all other SBSs are in the ``seller'' mode. In this case, only one coalition \{2,6\} is formed and a low system utility gain is realized. By contrast, in time slot 15, a more balanced composition of ``seller'' and ``buyers'' is presented, where buyers have a large demand for computation resource which can be potentially satisfied by the sellers. This motivates the formation of SBS coalitions and hence results in larger system utility gain.

\section{Related work}\label{sec_related_work}
Resource allocation in small cell networks has been widely studied in the literature. For instance, belief propagation method is applied to manage interference problems in the femtocell network\cite{rangan2012belief}. Orthogonal bandwidth allocation for femtocells is investigated using fractional frequency reuse \cite{lee2011interference}. However, these works only manage communication resources and do not consider the computing capabilities of SBSs and how SBSs can form a pool of computational resources while providing radio access service. Computational resource sharing is the main concern of geographical load balancing techniques originally proposed for data centers to deal with spatial diversities of workload patterns \cite{lin2012online} and electricity prices \cite{lou2015spatio}. However, conventional clouds manage only computational resources without considering the radio access. In edge systems, the provisioning of radio access between the SBSs and end users is a critical design aspect that has a significant impact on the system performance. Recently, with SBSs being considered as a major form of edge devices in the new edge computing paradigm, increasingly many works start to look at the joint optimization of computation and communication resources of SBSs \cite{munoz2015optimization, barbarossa2014computation}. However, most of these works assume that the coalitions among the SBSs have already been determined and focus on the offloading decisions of MUEs. In \cite{queis2015fog, queis2014smallcell}, computation clusters of SBSs are optimally built by performing joint allocation of radio and computation resources, focusing on reducing the power consumption and processing complexity. However, the incentive issue of SBSs is not considered.

Since SBSs are typically developed by individual owners, how to provide SBSs with incentives to collaborate is key to improving the overall system performance. In \cite{chen2012utility,yun2012economic,langar2015operations}, incentive mechanisms based on game theoretic methods are designed to motivate femtocells to adopt open/hybrid radio access mode, which does not consider the cooperative edge computing.  In \cite{jin2016auction}, a double auction mechanism is developed for matching user computation request to a set of cloudlets. These works consider incentives of SBSs to provide services to end users. Our paper studies incentives of SBSs to collaborate among each other.

The coalitional game theory is a powerful tool for studying user cooperation problems with individual incentive issues, which has been widely adopted in various problems, e.g. interference management \cite{zhang2014coalitional,pantisano2013interference}, spectrum sensing \cite{saad2011spectrum} and smart grids \cite{saad2011coalotional,lee2014direct}. Recently, it is also introduced in \cite{tanzil2016distributed} to the collaborative edge computing system for forming Femto-clouds among individual femtocells. Our considered system differs from this paper in the following aspects. First, we consider an ultra dense deployment scenario. Therefore, collaboration does not only occur on the SBS peer offloading level but also on the MUE-SBS association level. Second, we consider socially-trusted collaboration by introducing the social trust network. This allows risk management between any pair of SBSs.

\section{Conclusion}\label{sec_conclusion}
In this paper, we investigated collaborative edge computing in a densely-deployed small cell network, where the SBSs are incentivized to form coalitions and cooperate with SBS peers to increase system utility. Within the coalitions, buyers and sellers are allowed to cooperate at both MUE-SBS association stage and SBS peer offloading stage by exploiting the dense deployment of SBSs. We developed a distributed coalition formation algorithm based on Merge-and-Split rules. The proposed algorithm jointly considers SBS operational cost, cloud service fees, and potential security risks during coalition formation and guarantees the stability of coalitions by following a payment-based incentive mechanism. Our simulation results show that the collaborative edge system can dramatically reduce the system cost by more than 40\%. Our study not only provides new guidelines for cooperation among densely deployed edge devices but also delivers important insights for social trust and security issues in edge cooperation, which deserves more future investigation in collaborative edge computing systems.

\bibliographystyle{IEEEtran}
\bibliography{refs}

\end{document}